\documentclass[runningheads,a4paper]{llncs}
    \usepackage{amsmath}
    \usepackage{amsxtra}
    \usepackage{amstext}
    \usepackage{amssymb}
    \usepackage{latexsym}
    \usepackage{dsfont} 
\newcommand{\GF}[1]{{\mathbb F}_{#1}}

\renewenvironment{proof}{\begin{oldproof}}{\qed\end{oldproof}}

\begin{document}
\title{Solving Some Affine Equations over Finite Fields}
\author{ Sihem Mesnager\inst{1} \and  Kwang Ho Kim\inst{2,3} \and  Jong Hyok Choe\inst{2} \and Dok Nam Lee\inst{2} } \institute{ Department of Mathematics, University of Paris VIII, 93526 Saint-Denis, France, University of Paris XIII, CNRS, LAGA UMR 7539, Sorbonne Paris Cit\'e, 93430 Villetaneuse,
France, and Telecom ParisTech 75013 Paris. \email{Email:
smesnager@univ-paris8.fr} \and Institute of Mathematics, State
Academy of Sciences, Pyongyang, Democratic People's Republic of
Korea. \email{Email: khk.cryptech@gmail.com} \and PGItech Corp.,
Pyongyang, Democratic People's Republic of Korea.}

\authorrunning{S. Mesnager, K.H. Kim, J.H. Choe and D.N. Lee}

{\renewcommand{\thefootnote}{}\footnotetext{Version : \today}}
\maketitle

\begin{abstract}
Let $l$ and $k$ be two integers such that $l|k$. Define
$T_l^k(X):=X+X^{p^l}+\cdots+X^{p^{l(k/l-2)}}+X^{p^{l(k/l-1)}}$ and
$S_l^k(X):=X-X^{p^l}+\cdots+(-1)^{(k/l-1)}X^{p^{l(k/l-1)}}$, where
$p$ is any prime.

This paper gives explicit representations of all solutions in
$\GF{p^n}$ to the affine equations $T_l^{k}(X)=a$ and
$S_l^{k}(X)=a$, $a\in \GF{p^n}$. For the case $p=2$ that was solved
very recently in \cite{MKCL2019}, the result of this paper reveals
another solution.
\end{abstract}



\section{Introduction}
Let $\GF{p^n}$ be the finite field of $p^n$ elements where $p$ is a
prime and $n\geq 1$ is a positive integer. A polynomial
$L(X)\in\GF{p^n}[X]$ of shape
\begin{displaymath}
  L(X) = \sum_{i=0}^t a_i X^{p^i}, a_i\in\GF{p^n}
\end{displaymath}
is called a \emph{linearized polynomial} over $\GF{p^n}$. An affine
equation  over $\GF{p^n}$ is an equation of type
\begin{equation}\label{Affine}
L(X)=a,
\end{equation}
where $L$ is a linearized polynomial and $a\in \GF{p^n}$.

Affine equations arise in many several different problems and
contexts such as rank metric codes and linear sets. In particular,
those involving the trace functions are crucial in many contexts of
cryptography and error-correcting codes
\cite{BSS1999,CarletBook,CarletBook1}. Recent research on linearized polynomials and related topics can be found in \cite{CSAJBOK2020,CSAJBOK2019109,MKCL2019,Olga-Zullo2019,Zanella2019}. However, to explicit the
solutions of affine equations
 is often challenging as the ultimate goal of such
study.

In this paper, we study the following two affine equations
\begin{eqnarray}
  \label{equation:T}  &&T_l^k(X) := \sum_{i=0}^{\frac kl-1} X^{p^{li}}=a,\\
  \label{equation:S}  &&S_l^k(X) := \sum_{i=0}^{\frac kl-1} (-1)^i X^{p^{li}}=a,
\end{eqnarray}
 where $a\in \GF{p^n}$, $k$ and $l$ are positive integers such that $l\vert k$.

It is well-known that a linearized polynomial induce a linear
transformation of $\GF{p^n}$ over $\GF{p}$. In particular, if $x_1$
and $x_2$ are two solutions in $\GF{p^n}$ to
Equation~(\ref{Affine}), then their difference $x_1-x_2$ is a zero
of $L$ in $\GF{p^n}$, that is, their difference lies in the set
$\{x\in\GF{p^n}\mid L(x) = 0\}$, that we call the kernel of $L$
restricted to $\GF{p^n}$. Determination of the $\GF{p^n}-$ solutions
to Equation~(\ref{Affine}) can therefore be divided into two
problems: determining the kernel of $L$ restricted to $\GF{p^n}$ and
explicit a particular solution $x_0$ in $\GF{p^n}$. Indeed, if
those two problems are solved then the set of all
$\GF{p^n}-$solutions
 to Equation~(\ref{Affine}) is $x_0+\{x\in\GF{p^n}\mid
L(x) = 0\}$.

In this paper, we solve these two problems for the linearized
polynomials $T^k_l$ and $S^k_l$. We firstly determine the kernels of
$T^k_l$ and $S^k_l$ in Section~\ref{subsection:kernels}. Next, we
give explicit representations of particular solutions to
Equation~\eqref{equation:T} and Equation~\eqref{equation:S} in
Section~\ref{subsection:roots}. As by-product of those results, we
also characterize the elements $a$ in $\GF{p^n}$ for which
Equation~\eqref{equation:T} and Equation~\eqref{equation:S} has at
least one solution in $\GF{p^n}$ in Section~\ref{subsection:roots}.

\begin{remark} In \cite{MKCL2019}, we considered the particular
case of $p=2$ for which $T_l^k(X)=S_l^k(X).$ Interestingly,
Theorem~\ref{ker_T} and Theorem~\ref{T_kneq2} in this paper provides
another solution for this particular case.
\end{remark}

\section{Preliminaries}
\label{sec:solving}
Throughout this paper, we maintain the following notations (otherwise, we will point out it at the appropriate place).\\
\textbullet\quad $p$ is any prime and $n$ is a positive integer. \\
\textbullet\quad $a$ is any element of the finite field
$\GF{p^n}$.\\
\textbullet\quad $k$ and $l$ are positive integers such that $l\vert
k$. \\
\textbullet\quad We denote the greatest common divisor
 and the smallest common multiple of two positive integers $u$ and $v$ by $(u,v)$ and
 $[u,v]$,
respectively. \\
\textbullet\quad $d:=(n,k)$, $e:=(n,l)$ and $L:=[d,l]$. \\

We now present a lemma that will be frequently used throughout this
paper.

\begin{lemma}\label{lem_properties}
  For any positive integers $k$, $l$ and $m$ with $m\vert l\vert k$.
\begin{enumerate}
\item \(T_l^k\circ T_m^l(X)=T_m^k(X)\) is an identity. \(T_l^k\circ S_m^l(X)=S_m^k(X)\) if $l/m$ is even and
  \( S_l^k\circ S_m^l(X)=S_m^k(X)\) if $l/m$ is odd.
\item\label{lem_properties:2} \(S_l^k\circ T_l^{2l}(X)=S_k^{2k}(X)=X-X^{p^k}\) if
  $\frac{k}{l}$ is even and   \(S_l^k\circ T_l^{2l}(X)=T_k^{2k}(X)=X+X^{p^k}\) if $\frac{k}{l}$ is odd.
\item \label{lem_properties:3}  \(T_l^k\circ S_l^{2l}(X) =S_k^{2k}(X)\).
\item  \( T_k^{[n,k]}(x)=T_{d}^n(x)\) for
  any $x\in\GF{p^n}$. Furthermore, if
  $\frac{[n,k]}{k}$ is even, then \( S_k^{[n,k]}(x)=S_{d}^n(x)\) for
  any $x\in\GF{p^n}$.
\item   If $\frac{k}{l}$ is even, then \(S_l^k(x)+ S_l^k(x)^{p^l}=0\) for any \( x\in \GF{p^k} \).
\end{enumerate}
\end{lemma}
\begin{proof}
 The first three items are obtained by easy straightforward
  calculations. Hence, we give proofs only for the last two items.

  Since $nk=[n,k]d$, one has
  \(\frac{n}{d}=\frac{[n,k]}{k}\) and
  $\{jd \mid 0\leq j\leq \frac n{d}-1\} = \{ik\mod
  n \mid 0\leq i\leq\frac{[n,k]}{k}-1\}$ because $n$ divides $ik$ if
  and only if $i$ is a multiple of $\frac{n}{d}=\frac{[n,k]}{k}$. Therefore \( T_k^{[n,k]}(x)=T_{d}^n(x)\).

  Furthermore, if $\frac{[n,k]}{k}=\frac{n}{d}$ is even, then $\frac{k}{d}$ is odd since it is prime to $\frac{n}{d}$. Thus, when two integers $i$ and $j$ are such that $jd=ik \mod
  n$, they have the same parity. This proves $ S_k^{[n,k]}(x)=S_{d}^n(x)$ if
  $\frac{[n,k]}{k}$ is even.

  Finally, the last item is proved as follows: By Item 1,
  we have $S_l^k(X)=S_{l}^{2l}\circ T_{2l}^k(X)$. Let $x\in \GF{p^k}$. Then $y:=T_{2l}^k(x)\in
  \GF{p^{2l}}$. Hence $S_l^k(x)+ S_l^k(x)^{p^l}=S_{l}^{2l}(y)+{S_{l}^{2l}(y)}^{p^l}=(y-y^{p^l})+(y-y^{p^l})^{p^l}=0.$
\end{proof}

\section{On the Kernels of $T^k_l$ and $S^k_l$}
\label{subsection:kernels}

To determine the kernels of $T^k_l$ and $S^k_l$ in $\GF{p^n}$, we
begin with determining the zeros of $T^k_l$ and $S^k_l$ in the
algebraic closure $\overline{\GF{p}}$.

\begin{lemma}\label{T_l^k=0}
  It holds:
\begin{enumerate}
\item
\[
\{x\in
\overline{\GF{p}}\,|\,T_l^k(x)=0\}=S_l^{2l}(\GF{p^k})=\{x-x^{p^l}\,|\,x\in
\GF{p^k}\}.
\]
\item
When $\frac {k}{l}$ is even,
\[
\{x\in
\overline{\GF{p}}\,|\,S_l^k(x)=0\}=T_l^{2l}(\GF{p^k})=\{x+x^{p^l}\,|\,x\in
\GF{p^k}\}.
\]
\item
When $\frac {k}{l}$ is odd,
\[
\{x\in \overline{\GF{p}}\,|\,S_l^k(x)=0\}=S_k^{2k}\circ
T_l^{2l}(\GF{p^{2k}})=\{(x+x^{p^l})-(x+x^{p^l})^{p^k}\,|\,x\in
\GF{p^{2k}}\}.
\]

\end{enumerate}
\end{lemma}
\begin{proof}
   For $x\in \GF{p^k}$, by Item~\ref{lem_properties:3} of Lemma
  \ref{lem_properties}, $T_l^k\circ S_l^{2l}(x)=x-x^{p^k}=0$ proving
  the inclusion of   $S_l^{2l}(\GF{p^k})$ in $\{x\in \overline{\GF{p}}\,|\,T_l^k(x)=0\}$. We then conclude the equality from the fact that the
  two sets have the same cardinality $p^{k-l}$. The second and third items are similarly proved by using Item~\ref{lem_properties:2} of
  Lemma~\ref{lem_properties}.
\end{proof}

Based on the above lemma, we then deduce the kernels of $T^k_l$ and
$S^k_l$ restricted to $\GF{p^n}$. For the reader's convenience, we
present our results by three statements each corresponding to an
item of Lemma~\ref{T_l^k=0}.

\begin{theorem}\label{ker_T}
  The following holds true:
  \begin{displaymath}
    \{x\in \GF{p^n}\,|\,T_l^k(x)=0\}
    =\left\{\begin{array}{ll}
     \GF{p^{d}}, & \mbox{if $p\vert \frac{k}{L}$}\\
     S_{e}^{2e}(\GF{p^{d}}), & \mbox{otherwise.}
    \end{array}\right.
  \end{displaymath}
  Consequently,
    \begin{displaymath}
   \# \{x\in \GF{p^n}\,|\,T_l^k(x)=0\}
    =\left\{\begin{array}{ll}
    p^{d}, & \mbox{if $p\vert \frac{k}{L}$}\\
    p^{d-e}, & \mbox{otherwise.}
    \end{array}\right.
  \end{displaymath}
\end{theorem}
\begin{proof}
    By Item 1 of Lemma~\ref{T_l^k=0},
  \begin{displaymath}
    \{x\in \GF{p^n}\,|\,T_l^k(x)=0\}=S_l^{2l}(\GF{p^k})\cap \GF{p^n} \subset \GF{p^{(n,k)}}=\GF{p^d}.
  \end{displaymath}
 Therefore, by Item 1 of
  Lemma~\ref{lem_properties}
  \begin{eqnarray*}
    \{x\in \GF{p^n}\,|\,T_l^k(x)=0\}&=&\{x\in \GF{p^{d}}\,|\,T_l^k(x)=0\}\\
&=&\{x\in \GF{p^{d}}\,|\,\frac{k}{L}T_l^{L}(x)=0\}.
\end{eqnarray*}
Thus, if $p|\frac{k}{L}$, then
\begin{equation}\nonumber
\{x\in \GF{p^n}\,|\,T_l^k(x)=0\}=\GF{p^{d}},
\end{equation}
and if $p\nmid\frac{k}{L}$, then
\begin{align*}
  \{x\in\GF{p^n}\,|\,T_l^k(x)=0\}&=\{x\in
                                   \GF{p^{d}}\,|\,T_l^{L}(x)=0\}
\\
&=\{x\in \GF{p^{d}}\,|\,T_{e}^{d}(x)=0\} \text{ (by Item
4 of Lemma~\ref{lem_properties})}\\
                                 &=S_{e}^{2e}(\GF{p^{d}}) \text{ (by Item
1 of Lemma~\ref{T_l^k=0})}.
\end{align*}
\end{proof}

\begin{theorem}\label{ker_S_even}
Suppose that $\frac{k}{l}$ is even.
\begin{enumerate}
\item
If $\frac{d}{e}$ is even, then
\[
  \{x\in \GF{p^n}\,|\,S_l^k(x)=0\}=
  \left\{\begin{array}{ll}
           \GF{p^{d}}, & \mbox{if $p|\frac{k}{L}$} \\
           T_{e}^{2e}(\GF{p^{d}}), & \mbox{otherwise}
  \end{array}\right.
\]
and consequently
\[
  \#\{x\in \GF{p^n}\,|\,S_l^k(x)=0\}=
  \left\{\begin{array}{ll}
           p^{d}, & \mbox{if $p|\frac{k}{L}$} \\
           p^{d-e}, & \mbox{otherwise.}
  \end{array}\right.
\]
\item If $\frac{d}{e}$ is odd, then
\[
\{x\in \GF{p^n}\,|\,S_l^k(x)=0\}=\GF{p^{d}}
\]
and consequently
\[
\#\{x\in \GF{p^n}\,|\,S_l^k(x)=0\}=p^{d}.
\]
\end{enumerate}
\end{theorem}
\begin{proof}
By Item 2 of Lemma~\ref{T_l^k=0}, when $\frac{k}{l}$ is even, we
know that $\{x\in
\GF{p^n}\,|\,S_l^k(x)=0\}=T_l^{2l}(\GF{p^k})\cap\GF{p^n}\subset
\GF{p^{d}}$  and thus
\begin{align*}
\{x\in
\GF{p^n}\,|\,S_l^k(x)=0\}=\{x\in\GF{p^{d}}\,|\,S_l^{k}(x)=0\}.
\end{align*}

 Now, suppose that $\frac{d}{e}=\frac{d}{(d,l)}=\frac{L}{l}$ is even.
 Then, by Item 1 of Lemma~\ref{lem_properties}
\begin{align*}
\{x\in
\GF{p^{d}}\,|\,S_l^k(x)=0\}&=\{x\in\GF{p^{d}}\,|\,T^k_{L}\circ S_l^{L}(x)=0\}\\
&=\{x\in\GF{p^{d}}\,|\,\frac{k}{L}S_l^{L}(x)=0\}.
\end{align*}
\\Therefore, if $p|\frac{k}{L}$, then
\begin{equation}\nonumber
\{x\in \GF{p^n}\,|\,S_l^k(x)=0\}=\GF{p^{d}},
\end{equation}
 and if $p\nmid\frac{k}{L}$, then
\begin{align*}
\{x\in\GF{p^n}\,|\,S_l^k(x)=0\}&=\{x\in
\GF{p^{d}}\,|\,S_l^{L}(x)=0\}\\
&=\{x\in \GF{p^{d}}\,|\,S_{e}^{d}(x)=0\} \text{
(by Item 4 of Lemma~\ref{lem_properties})}\\
&=T_{e}^{2e}(\GF{p^{d}})\text{ (by Item 2 of Lemma~\ref{T_l^k=0})}.
\end{align*}

Suppose now that $\frac{d}{e}=\frac{L}{l}$ is odd. In this case,
$\frac{k}{L}$ is even as $\frac{k}{l}=\frac{k}{L}\cdot \frac{L}{l}$
is even by the assumption. Thus, we have
\begin{align*}
\{x\in \GF{p^n}\,|\,S_l^k(x)=0\}&=\{x\in
\GF{p^{d}}\,|\,S_l^k(x)=0\}\\
&=\{x\in \GF{p^{d}}\,|\,S_{L}^k\circ S_l^{L}(x)=0\} \text{
(by Item 1 of Lemma~\ref{lem_properties})}\\
&=\GF{p^{d}}
\end{align*}
because $S_l^{L}(x)\in\GF{p^{d}}\subset \GF{p^L}$ for
$x\in\GF{p^{d}}$.
\end{proof}

\begin{theorem}\label{ker_S_odd}
  Suppose that $\frac kl$ is odd.
\begin{enumerate}
\item
When $\frac{n}{d}$ is odd,
\[
\{x\in \GF{p^n}\,|\,S_l^k(x)=0\}=\{0\}
\]
and consequently
\[
\#\{x\in \GF{p^n}\,|\,S_l^k(x)=0\}=1.
\]
\item
When $\frac{n}{d}$ is even,
\[
\{x\in
\GF{p^n}\,|\,S_l^k(x)=0\}=
\left\{\begin{array}{ll}
         S_{d}^{2d}(\GF{p^{2d}}), & \mbox{if $p|\frac{k}{L}$} \\
        S_{d}^{2d}\circ T_{e}^{2e}(\GF{p^{2d}}), & \mbox{otherwise}
       \end{array}\right.
\]
and consequently
\[
\#\{x\in \GF{p^n}\,|\,S_l^k(x)=0\}= \left\{\begin{array}{ll}
         p^{d}, & \mbox{if $p|\frac{k}{L}$} \\
         p^{d-e}, & \mbox{otherwise.}
       \end{array}\right.
\]
\end{enumerate}
\end{theorem}

\begin{proof}  First of all, note that $\frac{L}{l}$ and $\frac{k}{L}$ are odd being divisors of
odd $\frac{k}{l}$ and
 by Item 3 of Lemma~\ref{T_l^k=0} one has
\begin{equation}\nonumber
\{x\in\GF{p^n}\,|\,S_l^k(x)=0\}=S_k^{2k}\circ
T_l^{2l}(\GF{p^{2k}})\cap \GF{p^n}=\{x\in
\GF{p^{(n,2k)}}\,|\,S_l^{k}(x)=0\}.
\end{equation}

Suppose that $\frac{n}{d}$ is odd. Then, $(n, 2k)=d$
 and we have
\begin{align*}
\{x\in \GF{p^n}\,|\,S_l^k(x)=0\}&=\{x\in
\GF{p^{d}}\,|\,S_l^k(x)=0\}\\
&=\{x\in
\GF{p^{d}}\,|\,S^k_{L}\circ S_l^{L}(x)=0\} \text{ (by Item 1 of Lemma~\ref{lem_properties})}\\
&=\{x\in \GF{p^{d}}\,|\, S_l^{L}(x)=0\}\\
&=\{S_{L}^{2L}\circ T_l^{2l}(\beta)\in
\GF{p^{d}}\,|\,\beta\in\GF{p^{2L}}\}  \text{ (by Item 3 of
Lemma~\ref{T_l^k=0})}.
\end{align*}
Now, if $S_{L}^{2L}\circ T_l^{2l}(\beta)\in \GF{p^{d}}$ for
$\beta\in\GF{p^{2L}}$, then we have
\begin{align*}
(S_{L}^{2L}\circ T_l^{2l}(\beta))^{p^{L}}=S_{L}^{2L}\circ
T_l^{2l}(\beta)&\Longleftrightarrow -S_{L}^{2L}\circ
T_l^{2l}(\beta)=S_{L}^{2L}\circ T_l^{2l}(\beta)\\
&\Longleftrightarrow S_{L}^{2L}\circ T_l^{2l}(\beta)=0.
\end{align*}
 Thus, in that case
$$\{x\in \GF{p^n}\,|\,S_l^k(x)=0\}=\{0\}.$$

Now, suppose that $\frac{n}{d}$ is even. Then, $(n, 2k)=2d$ and
$$\{x\in \GF{p^n}\,|\,S_l^k(x)=0\}=\{x\in
\GF{p^{2d}}\,|\,S_l^k(x)=0\}.$$
 Since $\frac{L}{l}$ is odd, from Item 5 of Lemma~\ref{lem_properties} it follows
\begin{equation}\label{eq1}
 S_l^{2L}(x)+S_l^{2L}(x)^{p^{L}}=0
\end{equation}
 for every $x\in\GF{p^{2d}}.$ And,
 $\frac{k}{d}$ and $\frac{L}{d}$ (being a divisor of $\frac{k}{d}$) are
 odd since  $\frac{n}{d}$ is even and $(\frac{n}{d}, \frac{k}{d})=1$. Therefore,  for every $x\in\GF{p^{2d}}$, it holds $$x^{p^L}=(x^{p^{\frac{L-d}{d}\cdot
 d}})^{p^d}=x^{p^d},$$
 $$x^{p^k}=(x^{p^{\frac{k-d}{d}\cdot
 d}})^{p^d}=x^{p^d}$$
 and hence
\begin{equation}\label{eq2}
S_{L}^{2L}(x)=S_{d}^{2d}(x) \text{ and } T_k^{2k}(x)=T_d^{2d}(x).
\end{equation}
Moreover, since $\frac{L}{d}=\frac{[d,l]}{d}$ is odd, one has $(2d,
l)=(d,l)=e$ and
$[2d,l]=\frac{2dl}{(2d,l)}=2\frac{dl}{(d,l)}=2[d,l]=2L$. Therefore,
by Item 4 of Lemma~\ref{lem_properties}, for every $x\in\GF{p^{2d}}$
we have
$S_l^{2L}(x)=S_l^{[2d,l]}(x)=S_{(2d,l)}^{2d}(x)=S_{(d,l)}^{2d}(x)=S_{e}^{2d}(x)$,
that is,
\begin{equation}\label{eq3}
S_l^{2L}(x)=S_{e}^{2d}(x).
\end{equation}
Then, for every $x\in\GF{p^{2d}}$ one has
\begin{align*}
S_{d}^{2d}\circ S_l^k(x)&=S_{L}^{2L}\circ
S_l^k(x)\text{ (by \eqref{eq2})}\\
&=S_{L}^{2L}\circ S_{L}^{k}\circ
S_l^{L}(x) \text{ (by Item 1 of Lemma~\ref{lem_properties})}\\
&=S_{L}^{k}\circ S_{L}^{2L}\circ
S_l^{L}(x)\\
&=S_{L}^{k}\circ S_{l}^{2L}(x) \text{ (again by Item 1 of Lemma~\ref{lem_properties})}\\
&=\frac{k}{L} S_{l}^{2L}(x)\text{ (by \eqref{eq1})}\\
&=\frac{k}{L}S_{e}^{2d}(x) \text{ (by \eqref{eq3})},
\end{align*}
that is,
\begin{equation}\label{eq4}
S_l^k(S_{d}^{2d}(x))=\left\{\begin{array}{ll}
         0, & \mbox{if $p|\frac{k}{L}$} \\
         S_{e}^{2d}(x), & \mbox{otherwise.}
       \end{array}\right.
\end{equation}
Thus, when $p|\frac{k}{L}$, it holds
$$
\{x\in\GF{p^{2d}}\,|\,S_l^k(x)=0\}\supset
S_{d}^{2d}(\GF{p^{2d}})=\ker(T_d^{2d}),
$$
where the equality is from Item 1 of Lemma~\ref{T_l^k=0}. By the
way, by \eqref{eq2} and  Item 2 of Lemma~\ref{lem_properties},
$T_d^{2d}(x)=T_k^{2k}(x)=T_l^{2l}\circ S_l^k(x)$ for every
$x\in\GF{p^{2d}}$ and therefore
$\{x\in\GF{p^{2d}}\,|\,S_l^k(x)=0\}\subset \ker(T_d^{2d}).$ Hence,
when $p|\frac{k}{L}$, we get
\begin{equation*}
\{x\in\GF{p^n}\,|\,S_l^k(x)=0\}=S_{d}^{2d}(\GF{p^{2d}}).
\end{equation*}

On the other hand, if $p\nmid\frac{k}{L}$, then by \eqref{eq4}
\begin{align*}
\{x\in\GF{p^{2d}}\,|\,S_l^k(x)=0\}&\supset
\{S_d^{2d}(x)\,|\,S_{e}^{2d}(x)=0, x\in\GF{p^{2d}}\}\\
&=\{S_d^{2d}\circ T_{e}^{2e}(\beta)\,|\,\beta\in\GF{p^{2d}}\} \text{
(by Item 2 of Lemma~\ref{T_l^k=0})}\\
&=\ker(S_e^d)\text{ (by Item 3 of Lemma~\ref{T_l^k=0} since
$\frac{d}{e}=\frac{L}{l}$ is odd)}.
\end{align*}
By the way, when $x\in\GF{p^{2d}}$, we can write
$$S_l^k(x)=\frac{k-L}{2L}S_l^{2L}(x)+S_l^L(x)$$
because $\frac{k}{L}$ is odd, and therefore if $S_l^k(x)=0$ for
$x\in\GF{p^{2d}}$, then $S_l^L(x)=0$ (since $S_l^{2L}(x)=0$ by
\eqref{eq3} and \eqref{eq4} when $S_l^k(x)=0$), that is, one has
$$\{x\in\GF{p^{2d}}\,|\,S_l^k(x)=0\}\subset \ker(S_{l}^{L})\cap \GF{p^{2d}}.$$ Thus, to conclude the
theorem it is sufficient to show:
\begin{equation}\label{eq5}
\ker(S_e^d)=\ker(S_{l}^{L})\cap \GF{p^{2d}}.
\end{equation}

To begin with, let us show $$\ker(S_e^d)\subset \ker(S_{l}^{L})\cap
\GF{p^{2d}}.$$ In fact, if $y\in \ker(S_e^d)$ or equivalently $
y=S_d^{2d}\circ T_{e}^{2e}(\beta)$ for some $\beta\in \GF{p^{2d}}$,
then
\begin{align*}
S_l^L(y)&=S_l^L\circ S_d^{2d}\circ T_{e}^{2e}(\beta)\\
&=S_l^L\circ S_L^{2L}\circ T_{e}^{2e}(\beta) \text{ (by \eqref{eq2})}\\
&=S_l^{2L}\circ T_{e}^{2e}(\beta)\text{ (by Item 1 of
Lemma~\ref{lem_properties})}\\
&=S_e^{2d}\circ T_{e}^{2e}(\beta)\text{ (by \eqref{eq3})}\\
&=S_{2d}^{4d}(\beta) \text{ (by Item 2 of
Lemma~\ref{lem_properties})}\\&=0 \text{ (since  $\beta\in
\GF{p^{2d}}$)}.
\end{align*}

Next, we prove $$\#\ker(S_e^d)= \#\{\ker(S_{l}^{L})\cap
\GF{p^{2d}}\}$$ which will conclude \eqref{eq5}. Let
$A:=\ker(S_{e}^{2d})=T_{e}^{2e}(\GF{p^{2d}})$. Then, by \eqref{eq3},
$A=\ker(S_{l}^{2L})\cap \GF{p^{2d}}$, and since
$S_{l}^{2L}=S_{L}^{2L}\circ S_l^L$ by Item 1 of
Lemma~\ref{lem_properties},
$$\ker(S_{l}^{L})\cap \GF{p^{2d}}\subset A.$$
Hence, now we determine $S_l^L(A)$ which will make us possible to
compute $\#\{\ker(S_{l}^{L})\cap \GF{p^{2d}}\}$. First, since
$S_d^{2d}(S_l^L(A))\overset{\eqref{eq2}}{=}S_L^{2L}(S_l^L(A))=S_l^{2L}(A)=\{0\}$,
it holds
\begin{equation}\label{eq6}
S_l^L(A)\subset \GF{p^d}.
\end{equation}
Then, since
$\frac{d}{e}=\frac{L}{l}$ is odd, by Item2 of
Lemma~\ref{lem_properties}, $S_{e}^d(A)=S_{e}^d\circ
T_{e}^{2e}(\GF{p^{2d}})=T_d^{2d}(\GF{p^{2d}})=\GF{p^d}$ and so
\begin{equation}\label{eq7}
\GF{p^d} \subset A.
\end{equation}
 Now, let us prove that $S_l^L$ is a permutation on $\GF{p^d}$. In fact, if
$y \in\GF{p^L}$ is an element in $\ker(S_l^L)$, then by Item 3 of
Lemma~\ref{T_l^k=0} we can write $y=S_L^{2L}\circ T_l^{2l}(\beta)$
for some $\beta\in \GF{p^{2L}}$ and one has
$y=y^{p^L}=(S_L^{2L}\circ T_l^{2l}(\beta))^{p^L}=-S_L^{2L}\circ
T_l^{2l}(\beta)=-y$, i.e. $y=0$. Therefore, $\ker(S_l^L)\cap
\GF{p^L}=\emptyset$ and $S_l^L$ is a permutation on $\GF{p^L}$ and
subsequently on $\GF{p^d}$. By acting $S_l^L$ on both sides of
\eqref{eq7} we get
\begin{equation}\label{eq8}
\GF{p^d} \subset S_l^L(A).
\end{equation}
Combined \eqref{eq6} and \eqref{eq8} proves $$S_l^L(A)=\GF{p^d}.$$
From here it follows $$\#\{\ker(S_{l}^{L})\cap
\GF{p^{2d}}\}=\#A/p^d= p^{(2d-e)-d}=p^{d-e}=\#\ker(S_e^d).$$
\end{proof}

\section{On particular solutions of (\ref{equation:T}) and
  (\ref{equation:S})}
\label{subsection:roots}

In this section, we give particular solutions for each of the two
equations (\ref{equation:T}) and (\ref{equation:S}) as well as
characterizations of the $a$'s in $\GF{p^n}$ for which the equations
(\ref{equation:T}) and (\ref{equation:S}) have at least one solution
in $\GF{p^n}$. We present these results by three statements  as in
Section~\ref{subsection:kernels} each of them corresponding again to
the three cases defined in Lemma~\ref{T_l^k=0}.

\begin{theorem}\label{T_kneq2} Let $\delta\in \GF{p^n}^*$ and $\delta_1\in \GF{p^d}^*$ be any elements such that
$T_{d}^n(\delta)=1$ and  $T_{e}^d(\delta_1)=1$.
\begin{enumerate}
\item  When $p|\frac{k}{L}$, there exists a solution in $\GF{p^n}$ to
the equation $T_l^k(X)=a$ if and only if $T_{d}^n(a)=0$. In that
case,
\begin{equation}\nonumber
x_0=S_l^{2l}(\sum_{i=0}^{\frac{n}{d}-2}\sum_{j=i+1}^{\frac{n}{d}-1}\delta^{p^{kj}}a^{p^{ki}})
\end{equation}
is a particular $\GF{p^n}-$solution to the equation $T_l^k(X)=a$.

\item When
$p\nmid\frac{k}{L}$, there exists a solution in $\GF{p^n}$ to the
equation $T_l^k(X)=a$ if and only if $S_{e}^{2e}\circ T_{d}^n(a)=0$.
In that case,
\begin{equation}\nonumber
  x_0=y_0+\frac{L}{k}(a-T_l^k(y_0))\delta_1,
\end{equation}
where
\begin{equation}\nonumber
y_0=\sum_{i=0}^{\frac{n}{d}-2}\sum_{j=i+1}^{\frac{n}{d}-1}\delta^{p^{kj}}S_l^{2l}(a)^{p^{ki}},
\end{equation}
is a particular $\GF{p^n}-$solution to the equation $T_l^k(X)=a$.
\end{enumerate}
\end{theorem}
\begin{proof}
  Let $a=T_l^k(x_0)$ for some $x_0\in\GF{p^n}$. Then
\begin{align*}
T_{d}^n(a)&=T_{d}^n\circ T_l^k(x_0)\\
&=T_{d}^n\circ
T_{L}^k\circ T_l^{L}(x_0) \text{ (by Item 1 of Lemma~\ref{lem_properties})}\\
&=T_{L}^k\circ T_l^{L}\circ T_{d}^n(x_0)\\
&=T_{L}^k\circ T_{e}^{d}\circ
T_{d}^n(x_0)\text{ (by Item 4 of Lemma~\ref{lem_properties})}\\
&=T_{L}^k\circ T_{e}^{n}(x_0)\text{ (by Item 1 of
Lemma~\ref{lem_properties})}\\&=\frac{k}{L}T_{e}^{n}(x_0) \text{
(since $T_{e}^{n}(x_0)\in \GF{p^e}\subset \GF{p^L}$)}.
\end{align*}
Thus, if $p|\frac{k}{L}$, then $T_{d}^n(a)=0$, and if
$p\nmid\frac{k}{L}$, then $S_{e}^{2e}\circ
T_{d}^n(a)=S_{e}^{2e}\circ T_{e}^{n}(x_0)=S_n^{2n}(x_0)= 0$ where we
applied Item 3 of Lemma~\ref{lem_properties}. In other words, if
$p|\frac{k}{L}$, then
\begin{equation}\label{im1}
T_l^k(\GF{p^n})\subset\{a\in\GF{p^n}\,|\,T_{d}^n(a)=0\},
\end{equation}
and if $p\nmid\frac{k}{L}$, then
\begin{equation}\label{im2}
T_l^k(\GF{p^n})\subset\{a\in\GF{p^n}\,|\,S_{e}^{2e}\circ
T_{d}^n(a)=0\}.
\end{equation}

By the way, by Theorem~\ref{ker_T} we have:
$$\#T_l^k(\GF{p^n})=p^n/\#\{\ker(T_l^k)\cap
\GF{p^n}\}=\left\{\begin{array}{ll}
    p^{n-d}, & \mbox{if $p\vert \frac{k}{L}$}\\
    p^{n-(d-e)}, & \mbox{otherwise.}
    \end{array}\right.$$
On the other hand, by the well-known nature of the trace mapping one
knows
$$\#\{a\in\GF{p^n}\,|\,T_{d}^n(a)=0\}=p^{n-d}$$ and
$$\#\{a\in\GF{p^n}\,|\,S_{e}^{2e}\circ
T_{d}^n(a)=0\}=\#\{a\in\GF{p^n}\,|\, T_{d}^n(a)\in
\GF{p^e}\}=p^{n-(d-e)}.$$ Thus, we conclude that the inclusions
\eqref{im1} and \eqref{im2} are indeed equalities. That is, the if
and only if conditions for $T_l^k(X)=a$ to have a
$\GF{p^n}-$solution are justified.

Let us check the validity of the given particular solutions. If
$T_d^n(a)=0$, we have
\begin{align*}
T_l^k\left(S_l^{2l}(\sum_{i=0}^{\frac{n}{d}-2}\sum_{j=i+1}^{\frac{n}{d}-1}\delta^{p^{kj}}a^{p^{ki}})\right)&=T_l^k\circ
S_l^{2l}(\sum_{i=0}^{\frac{n}{d}-2}\sum_{j=i+1}^{\frac{n}{d}-1}\delta^{p^{kj}}a^{p^{ki}})\\
&=S_k^{2k}(\sum_{i=0}^{\frac{n}{d}-2}\sum_{j=i+1}^{\frac{n}{d}-1}\delta^{p^{kj}}a^{p^{ki}}) \text{ (by Item 3 of Lemma~\ref{lem_properties})}\\
&=\sum_{i=0}^{\frac{n}{d}-2}\sum_{j=i+1}^{\frac{n}{d}-1}\delta^{p^{kj}}a^{p^{ki}}-\sum_{i=1}^{\frac{n}{d}-1}\sum_{j=i+1}^{\frac{n}{d}}\delta^{p^{kj}}a^{p^{ki}}\\
&=a-\delta T_d^n(a)=a.
\end{align*}
Now, suppose that $S_{e}^{2e}\circ T_{d}^n(a)=0$ i.e. $T_{d}^n(a)\in
\GF{p^e}$. Then, $S_l^{2l}(T^n_d(a))=T^n_d(a)-T^n_d(a)^{p^l}=0$, and
for
$y_0=\sum_{i=0}^{\frac{n}{d}-2}\sum_{j=i+1}^{\frac{n}{d}-1}\delta^{p^{kj}}S_l^{2l}(a)^{p^{ki}}$,
it holds
\begin{align*}
S_k^{2k}(y_0)&=\sum_{i=0}^{\frac{n}{d}-2}\sum_{j=i+1}^{\frac{n}{d}-1}\delta^{p^{kj}}S_l^{2l}(a)^{p^{ki}}-\sum_{i=0}^{\frac{n}{d}-1}\sum_{j=i+1}^{\frac{n}{d}}\delta^{p^{kj}}S_l^{2l}(a)^{p^{ki}}\\
&=S_l^{2l}(a)-\delta T^n_d(S_l^{2l}(a))=S_l^{2l}(a)-\delta
S_l^{2l}(T^n_d(a))\\
&=S_l^{2l}(a).
\end{align*}

Since $S_k^{2k}(y_0)=S_l^{2l}(T_l^k(y_0)) $ (Item 3 of
Lemma~\ref{lem_properties}), we get $$\beta:=a-T_l^k(y_0)\in
\ker(S_l^{2l})\cap \GF{p^n}\subset \GF{p^l}\cap
\GF{p^n}=\GF{p^{e}}.$$ Now, $\beta \delta_1\in \GF{p^d}\subset
\GF{p^L}$ and therefore
\begin{align*}
T_l^k(\beta \delta_1) &=T_L^k\circ T_l^L(\beta \delta_1)\text{ (by
Item 1 of Lemma~\ref{lem_properties})}\\&= \frac{k}{L}T_l^L(\beta
\delta_1) \text{ (since $T_l^L(\beta \delta_1)\in
\GF{p^L}$)}\\&=\frac{k}{L}T_e^d(\beta \delta_1) \text{ (by Item 4 of
Lemma~\ref{lem_properties})}\\
&=\frac{k}{L}\beta T_e^d( \delta_1) \text{ (since $\beta\in
\GF{p^e}$)}\\
&=\frac{k}{L}\beta.
\end{align*}
That is, we get
$T_l^k(\frac{L}{k}(a-T_l^k(y_0))\delta_1)=a-T_l^k(y_0)$, or
equivalently, $$T_l^k(y_0+\frac{L}{k}(a-T_l^k(y_0))\delta_1)=a.$$
\end{proof}

\begin{theorem}\label{S_kneq2-1} Let  $p\neq 2$, $\frac{k}{l}$
even, $\delta\in \GF{p^n}^*$ and $\delta_1\in \GF{p^d}^*$ be any
elements such that $T_{d}^n(\delta)=1$ and  $T_{2e}^d(\delta_1)=1$.
\begin{enumerate}
\item When $\frac{d}{e}$ is odd, or, when $\frac{d}{e}$ is even and
$p|\frac{k}{L}$, there exists a solution in $\GF{p^n}$ to the
equation $S_l^k(X)=a$ if and only if $T_{d}^n(a)=0$.  In that case,
\begin{equation}
x_0=T_l^{2l}(\sum_{i=0}^{\frac{n}{d}-2}\sum_{j=i+1}^{\frac{n}{d}-1}\delta^{p^{kj}}a^{p^{ki}})
\end{equation}
is a particular $\GF{p^n}-$solution to the equation $S_l^k(X)=a$.
\item When $\frac{d}{e}$ is even and $p\nmid\frac{k}{L}$, there exists
a solution in $\GF{p^n}$ to the equation $S_l^k(X)=a$ if and only if
$T_{e}^{2e}\circ T_{d}^n(a)=0$. In that case,
\begin{equation}\nonumber
x_0=y_0+\frac{L}{2k}(a-S_l^k(y_0))\delta_1,
\end{equation}
where
\begin{equation}\nonumber
y_0=\sum_{i=0}^{\frac{n}{d}-2}\sum_{j=i+1}^{\frac{n}{d}-1}\delta^{p^{kj}}T_l^{2l}(a)^{p^{ki}},
\end{equation}
is a particular $\GF{p^n}-$solution to the equation $S_l^k(X)=a$.
\end{enumerate}
\end{theorem}
\begin{proof}
Suppose that $S_l^k(x_0)=a$ for some $x_0\in \GF{p^n}$. When
$\frac{d}{e}=\frac{L}{l}$ is odd,  $\frac{k}{L}$ is even since
$\frac{k}{l}=\frac{L}{l}\cdot \frac{k}{L}$ was assumed to be even.
Then, we have
\begin{align*}
T_{d}^n(a)&=T_{d}^n\circ S_l^k(x_0)\\&=T_{d}^n\circ
S_{L}^{k}\circ S_l^{L}(x_0) \text{ (by Item 1 of Lemma~\ref{lem_properties})}\\
&=S_{L}^{k}( S_l^{L}\circ T_{d}^n(x_0))
\\&=0 \text{ (since
$S_l^{L}\circ T_{d}^n(x_0)\in \GF{p^d}\subset \GF{p^L}$ and
$\frac{k}{L}$ is even)}
\end{align*}
and thus
\begin{equation}\label{im3}
S_l^k(\GF{p^n})\subset \{a\in\GF{p^n}\,|\,T_{d}^n(a)=0\}.
\end{equation}

On the other hand, when $\frac{d}{e}=\frac{L}{l}$ is even, one has
\begin{align*}
T_{d}^n(a)&=T_{d}^n\circ S_l^k(x_0)\\
&=T_{d}^n\circ
T_{L}^{k}\circ S_l^{L}(x_0) \text{ (by Item 1 of Lemma~\ref{lem_properties})}\\
&=T_{L}^{k}\circ S_l^{L}\circ
T_{d}^n(x_0)\\
&=T_{L}^{k}\circ S_{e}^{d}\circ
T_{d}^n(x_0) \text{ (by Item 4 of Lemma~\ref{lem_properties})}\\
&=\frac{k}{L} S_{e}^{d}\circ T_{d}^n(x_0)
\text{ (since $S_{e}^{d}\circ T_{d}^n(x_0)\in \GF{p^d}\subset \GF{p^L}$)}.\\
\end{align*}
Therefore, if $p|\frac{k}{L}$, then \eqref{im3} still holds true,
and if $p\nmid\frac{k}{L}$, then it holds
\begin{align*}
T_{e}^{2e}\circ T_{d}^n(a)&=T_{e}^{2e}\circ S_{e}^{d}\circ
T_{d}^n(x_0)\\
&=S_d^{2d}\circ T_{d}^n(x_0)\text{ (by Item 2 of
Lemma~\ref{lem_properties})}\\&=0 \text{ (since $T_{d}^n(x_0)\in
\GF{p^d}$)}
\end{align*}
and thus
\begin{equation}\label{im4}
S_l^k(\GF{p^n})\subset \{a\in\GF{p^n}\,|\,T_{e}^{2e}\circ
T_{d}^n(a)=0\}.
\end{equation}
By the way, by Theorem~\ref{ker_S_even} we have:
$$\#S_l^k(\GF{p^n})=p^n/\#\{\ker(S_l^k)\cap
\GF{p^n}\}=\left\{\begin{array}{ll}
    p^{n-d}, & \mbox{if $\frac{d}{e}$ is odd, or, $\frac{d}{e}$ is even and $p\vert \frac{k}{L}$}\\
    p^{n-(d-e)}, & \mbox{otherwise.}
    \end{array}\right.$$
On the other hand, by the well-known nature of the trace mapping one
knows
$$\#\{a\in\GF{p^n}\,|\,T_{d}^n(a)=0\}=p^{n-d}$$ and
$$\#\{a\in\GF{p^n}\,|\,T_{e}^{2e}\circ
T_{d}^n(a)=0\}=p^{n-(d-e)}.$$ Therefore the inclusions \eqref{im3}
and \eqref{im4} are indeed equalities. That is, the if and only if
conditions for $S_l^k(X)=a$ to have a $\GF{p^n}-$solution are
justified.

Since  $S_l^k\circ T_l^{2l}=S_k^{2k}$ (Item 2 of
Lemma~\ref{lem_properties}), it can be checked  by the same
computation as in the proof of Item 1 of Theorem~\ref{T_kneq2} that
under the condition $T_d^n(a)=0$,
$$x_0=T_l^{2l}(\sum_{i=0}^{\frac{n}{d}-2}\sum_{j=i+1}^{\frac{n}{d}-1}\delta^{p^{kj}}a^{p^{ki}})$$
is a particular $\GF{p^n}-$solution to the equation $S_l^k(X)=a$.

Now, assuming that $\frac{d}{e}=\frac{L}{l}$ is even, let us suppose
that $T_{e}^{2e}\circ T_{d}^n(a)=0$ i.e.
$T_{d}^n(a)^{p^e}=-T_{d}^n(a)$. Then, $\frac{l}{e}$ is odd since it
is prime to $\frac{n}{e}=\frac{d}{e}\cdot\frac{n}{d}$ which is even.
Hence, $T_l^{2l}(T^n_d(a))=T^n_d(a)+T^n_d(a)^{p^l}=0$, and for
$y_0=\sum_{i=0}^{\frac{n}{d}-2}\sum_{j=i+1}^{\frac{n}{d}-1}\delta^{p^{kj}}T_l^{2l}(a)^{p^{ki}}$,
it holds
\begin{align*}
S_k^{2k}(y_0)&=\sum_{i=0}^{\frac{n}{d}-2}\sum_{j=i+1}^{\frac{n}{d}-1}\delta^{p^{kj}}T_l^{2l}(a)^{p^{ki}}-\sum_{i=0}^{\frac{n}{d}-1}\sum_{j=i+1}^{\frac{n}{d}}\delta^{p^{kj}}T_l^{2l}(a)^{p^{ki}}\\
&=T_l^{2l}(a)-\delta T^n_d(T_l^{2l}(a))=T_l^{2l}(a)-\delta
T_l^{2l}(T^n_d(a))\\
&=T_l^{2l}(a).
\end{align*}

Since $S_k^{2k}(y_0)=T_l^{2l}(S_l^k(y_0)) $ (Item 2 of
Lemma~\ref{lem_properties}), letting $\beta:=a-S_l^k(y_0)$, we have
$$\beta\in \ker(T_l^{2l})\cap \GF{p^n}\subset \GF{p^{2l}}\cap
\GF{p^n}\subset \GF{p^{2e}}\subset \GF{p^{d}}$$ and
\begin{align*}
S_l^k(\beta \delta_1) &= \frac{k}{L}S_l^L(\beta \delta_1)\text{
(since $\frac{L}{l}$ is even)}\\&=\frac{k}{L}S_e^d(\beta \delta_1)
\text{ (by Item 4 of
Lemma~\ref{lem_properties})}\\
&=\frac{k}{L}S_e^{2e}( T^d_{2e}( \beta\delta_1)) \text{ (by Item 1
of Lemma~\ref{lem_properties})}\\
&=\frac{k}{L}S_e^{2e}(\beta T^d_{2e}( \delta_1)) \text{ (since
$\beta \in \GF{p^{2e}}$)}\\
&=\frac{k}{L}S_e^{2e}(\beta)=\frac{k}{L}(\beta-\beta^{p^e}).
\end{align*}
On the other hand, since $\ker(T_l^{2l})\cap
\GF{p^{2e}}=S_e^{2e}(\GF{p^{2e}})$ (see Theorem~\ref{ker_T}),
$\beta\in \ker(T_l^{2l})\cap \GF{p^n}$ means that
$\beta=\alpha-\alpha^{p^e}$ for some $\alpha\in \GF{p^2e}$, and
therefore we get $
\beta+\beta^{p^e}=(\alpha-\alpha^{p^e})+(\alpha-\alpha^{p^e})^{p^e}=0$
and hence $$S_l^k(\beta \delta_1)=\frac{2k}{L}\beta,$$
 or
equivalently,
$$S_l^k\left(y_0+\frac{L}{2k}(a-S_l^k(y_0))\delta_1\right)=a.$$
\end{proof}

\begin{theorem}\label{S_kneq2-2} Let  $p\neq 2$,  $\frac{k}{l}$
odd, $\delta\in \GF{p^n}^*$ and $\delta_1\in \GF{p^d}^*$ be any
elements such that $T_{d}^n(\delta)=1$ and
$T_{2e}^{2d}(\delta_1)=1$.
\begin{enumerate}
\item When $\frac{n}{d}$ is even and $p|\frac{k}{L}$, there
exists a solution in $\GF{p^n}$ to the equation $S_l^k(X)=a$ if and
only if $S_{d}^n(a)=0$. In that case,
\begin{equation}\nonumber
x_0=T_l^{2l}(\sum_{i=0}^{\frac{n}{d}-2}\sum_{j=i+1}^{\frac{n}{d}-1}\delta^{p^{kj}}a^{p^{ki}}(-1)^i)
\end{equation}
is a particular $\GF{p^n}-$solution to the equation $S_l^k(X)=a$.

\item When $\frac{n}{d}$ is even and $p\nmid\frac{k}{L}$, there exists
a solution in $\GF{p^n}$ to the equation $S_l^k(X)=a$ if and only if
$T_{e}^{2e}\circ S_{d}^n(a)=0$. In that case,
\begin{equation}\nonumber
x_0=y_0+\frac{L}{2k}S_{d}^{2d}((a-S_l^k(y_0))\delta_1),
\end{equation}
where
\begin{equation}\nonumber
y_0=\sum_{i=0}^{\frac{n}{d}-2}\sum_{j=i+1}^{\frac{n}{d}-1}\delta^{p^{kj}}T_l^{2l}(a)^{p^{ki}}(-1)^i,
\end{equation}
is a particular $\GF{p^n}-$solution to the equation $S_l^k(X)=a$.

\item When $\frac{n}{d}$ is odd, the equation $S_l^k(X)=a$ has a unique $\GF{p^n}-$solution: $$x_0=\frac{T_l^{2l}\circ
S_k^{[n,k]}(a)}{2}.$$
\end{enumerate}
\end{theorem}
\begin{proof}
Suppose $a\in S_l^k(\GF{p^n})$ i.e. $a=S_l^k(x_0)$ for some
$x_0\in\GF{p^n}$.

Let us assume that $\frac{n}{d}$ is even. In this case,
 $\frac{k}{d}$ and its divisor $\frac{L}{d}$ are
 odd since $(\frac{n}{d}, \frac{k}{d})=1$. One has
\begin{align*}
S_{d}^n(a)&=S_{d}^n\circ S_l^k(x_0)\\&=S_{d}^n\circ
S_{L}^k\circ S_l^{L}(x_0) \text{ (by Item 1 of Lemma~\ref{lem_properties})}\\
&=S_{d}^{2d}\circ T_{2d}^n\circ
S_{L}^k\circ S_l^{L}(x_0) \text{ (again by Item 1 of Lemma~\ref{lem_properties})}\\
&=S_{L}^{2L}\circ T_{2d}^n\circ S_{L}^k\circ S_l^{L}(x_0) \text{
(by \eqref{eq2})}\\
&=S_{L}^k\circ S_{L}^{2L}\circ
S_l^{L}\circ T_{2d}^n(x_0)\\
&=S_{L}^k\circ S_l^{2L}\circ T_{2d}^n(x_0) \text{ (by Item 1 of Lemma~\ref{lem_properties})}\\
&=\frac{k}{L}S_l^{2L}\circ T_{2d}^n(x_0) \text{ (by \eqref{eq1})}\\
&=\frac{k}{L}S_e^{2d}\circ T_{2d}^n(x_0) \text{ (by \eqref{eq3})}\\
&=\frac{k}{L}S_{e}^{n}(x_0)=\frac{k}{L}S_{e}^{2e}\circ
T_{2e}^{n}(x_0) \text{ (once again by Item 1 of
Lemma~\ref{lem_properties})}.
\end{align*}
Therefore, if $p|\frac{k}{L}$, then it holds
\begin{equation}\label{im5}
S_l^k(\GF{p^n})\subset\{a\in\GF{p^n}\,|\,S_{d}^n(a)=0\},
\end{equation}
and if $p\nmid\frac{k}{L}$, then it holds
\begin{equation}\label{im6}
S_l^k(\GF{p^n})\subset\{a\in\GF{p^n}\,|\,T_{e}^{2e}\circ
S_{d}^n(a)=0\}
\end{equation}
since $T_{e}^{2e}\circ S_{d}^n(S_l^k(x_0))=T_{e}^{2e}\circ
S_{e}^{2e}\circ T_{2e}^{n}(x_0)=0$ by Lemma~\ref{T_l^k=0}.
Comparisons of set cardinalities regarding Item 2 of
Theorem~\ref{ker_S_odd} make us conclude that the inclusions
\eqref{im5} and \eqref{im6} are indeed equalities. That is, the if
and only if conditions for $S_l^k(X)=a$ to have a
$\GF{p^n}-$solution are justified.

If $S_d^n(a)=0$, then for
$x_0=T_l^{2l}(\sum_{i=0}^{\frac{n}{d}-2}\sum_{j=i+1}^{\frac{n}{d}-1}\delta^{p^{kj}}a^{p^{ki}}(-1)^i)$,
\begin{align*}
S_l^k\left(x_0\right)&=T_k^{2k}(\sum_{i=0}^{\frac{n}{d}-2}\sum_{j=i+1}^{\frac{n}{d}-1}\delta^{p^{kj}}a^{p^{ki}}(-1)^i)
\text{ (by Item 2 of Lemma~\ref{lem_properties})}\\
&=\sum_{i=0}^{\frac{n}{d}-2}\sum_{j=i+1}^{\frac{n}{d}-1}\delta^{p^{kj}}a^{p^{ki}}(-1)^i+\sum_{i=1}^{\frac{n}{d}-1}\sum_{j=i+1}^{\frac{n}{d}}\delta^{p^{kj}}a^{p^{ki}}(-1)^{i-1}\\
&=a-\delta S_d^n(a)=a.
\end{align*}

Now, suppose that $\frac{n}{d}=\frac{[n,k]}{k}$ is even and
$T_{e}^{2e}\circ S_{d}^n(a)=0$ (i.e.
$S_{d}^n(a)^{p^e}=-S_{d}^n(a)$). Then, $\frac{l}{e}$ is odd since it
is prime to $\frac{n}{e}=\frac{d}{e}\cdot\frac{n}{d}$ which is even,
and hence $T_l^{2l}(S_d^n(a))=S_d^n(a)+S_d^n(a)^{p^l}=0$. Thus, for
$y_0=\sum_{i=0}^{\frac{n}{d}-2}\sum_{j=i+1}^{\frac{n}{d}-1}\delta^{p^{kj}}T_l^{2l}(a)^{p^{ki}}(-1)^i$
one has
\begin{align*}
T_k^{2k}(y_0)&=\sum_{i=0}^{\frac{n}{d}-2}\sum_{j=i+1}^{\frac{n}{d}-1}\delta^{p^{kj}}T_l^{2l}(a)^{p^{ki}}(-1)^i+\sum_{i=1}^{\frac{n}{d}-1}\sum_{j=i+1}^{\frac{n}{d}}\delta^{p^{kj}}T_l^{2l}(a)^{p^{ki}}(-1)^{i-1}\\
&=T_l^{2l}(a)-\delta S_d^n(T_l^{2l}(a))=T_l^{2l}(a).
\end{align*}
Since $T_k^{2k}(y_0)=T_l^{2l}\circ S_l^k(y_0)$ (by Item 2 of
Lemma~\ref{lem_properties}), we have
$$\beta:=a-S_l^{k}(y_0)\in \ker(T_l^{2l})\cap \GF{p^n}\subset \GF{p^{2e}}\subset \GF{p^{2d}}.$$
Now,
\begin{align*}
S_l^k\left(S_{d}^{2d}(\beta\delta_1)\right)&=S_l^k\left(S_{L}^{2L}(\beta\delta_1)\right)
\text{ (by \eqref{eq2})}\\
&=S_L^k\circ S_l^L\left(S_{L}^{2L}(\beta\delta_1)\right) \text{ (by Item 1 of Lemma~\ref{lem_properties})}\\
&=S_L^k\circ \left(S_l^L\circ S_{l}^{2L}(\beta\delta_1)\right)=S_L^k\circ \left(S_{l}^{2L}(\beta\delta_1)\right) \text{ (by Item 1 of Lemma~\ref{lem_properties})}\\
&=\frac{k}{L} S_{l}^{2L}(\beta\delta_1) \text{ (by \eqref{eq1})}\\
&=\frac{k}{L} S_{e}^{2d}(\beta\delta_1) \text{ (by \eqref{eq3})}\\
&=\frac{k}{L} S_{e}^{2e} \circ T_{2e}^{2d}(\beta\delta_1) \text{ (by Item 1 of Lemma~\ref{lem_properties})}\\
&=\frac{k}{L} S_{e}^{2e} (\beta T_{2e}^{2d}(\delta_1))=\frac{k}{L}
S_{e}^{2e} (\beta)=\frac{k}{L} (\beta-\beta^{p^e}).
\end{align*}
By the way, since $\beta\in \ker(T_l^{2l})\cap
\GF{p^{2e}}=S_e^{2e}(\GF{p^{2e}})$ (Theorem~\ref{ker_T}), it holds
$\beta+\beta^{p^e}=0$ and thus we get
$$S_l^k\left(S_{d}^{2d}(\beta\delta_1)\right)=\frac{2k}{L} \beta.$$
 That is,
$$S_l^k(y_0+\frac{L}{2k}S_{d}^{2d}((a-S_l^k(y_0))\delta_1))=a.$$

If $\frac{n}{d}=\frac{[n,k]}{k}$ is odd, then by
Theorem~\ref{ker_S_odd}, $S_l^k$ is a permutation on $\GF{p^n}$ and
the equation $S_l^k(X)=a$ has a unique $\GF{p^n}-$solution. In fact,
by successive applications of Item 2 of Lemma~\ref{lem_properties}
we have
\begin{align*}
S_l^k(T_l^{2l}\circ S_k^{[n,k]}(a/2))=T_k^{2k}\circ S_k^{[n,k]}(a/2)
= T_{[n,k]}^{2[n,k]}(a/2)=a.\end{align*}
\end{proof}

\section{Conclusion}
We explicitly determined the sets of preimages of linearized
polynomials
\begin{align*}
&T_l^k(X) := \sum_{i=0}^{\frac kl-1} X^{p^{li}},\\
&S_l^k(X) := \sum_{i=0}^{\frac kl-1} (-1)^i X^{p^{li}},
\end{align*}
over ${\overline{\GF{p}}}$ and over the finite field $\GF{p^n}$ for
any characteristic $p$ and any integer $n\geq 1$. In particular,
another solution for the case $p=2$ that was solved very recently in
\cite{MKCL2019} was obtained by Theorem~\ref{ker_T} and
Theorem~\ref{T_kneq2}.

\end{document}